\documentclass[english,letterpaper,11pt,reqno]{amsart}

\usepackage{appendix}
\usepackage{amsbsy}
\usepackage{amsfonts}
\usepackage{amsmath}
\usepackage{amssymb}
\usepackage{amsthm}
\usepackage{graphicx}


\usepackage{ifthen}
\usepackage{textcomp}
\usepackage{abstract}

\usepackage{parskip}  %no indent; spaces between paragraphs
\usepackage[font=footnotesize]{caption}

\newcommand{\lra}{\longrightarrow}
\newcommand{\ra}{\rightarrow}
\newcommand{\inc}{\ensuremath{\lhook\joinrel\relbar\joinrel\rightarrow}}

\newcommand{\RR}{\mathbb{R}}

\newtheorem{theorem}{Theorem}
\newtheorem*{theorem*}{Theorem}
\newtheorem{lemma}{Lemma}

\newtheorem{corollary}{Corollary}

\newcommand{\beqa}{\begin{eqnarray}}
\newcommand{\beq}{\begin{equation}}
\newcommand{\eeqa}{\end{eqnarray}}
\newcommand{\eeq}{\end{equation}}

\newcommand\ip[3]{\langle {#1},{#2}\rangle_{#3}}
\newcommand{\vtheta}{\vartheta}
\newcommand{\vphi}{\varphi}
\newcommand{\crit}{c_{\rm crit}}

\begin{document}
\title[Blowup of Jang's equation near a spacetime singularity]{Blowup solutions of Jang's equation near a spacetime singularity}
\author{Amir Babak Aazami}
\email{amir.aazami@ipmu.jp}
\address{Kavli IPMU (WPI), University of Tokyo, 5-1-5 Kashiwanoha, Kashiwa-shi, Japan 277-8583}

\author{Graham Cox}
\email{ghcox@email.unc.edu}
\address{Department of
Mathematics, University of North Carolina, Phillips Hall CB \#3250, Chapel Hill, NC 27599}
\maketitle
\begin{abstract}
We study Jang's equation on a one-parameter family of asymptotically flat, spherically symmetric Cauchy hypersurfaces in the maximally extended Schwarzschild spacetime. The hypersurfaces contain apparent horizons and are parametrized by their proximity to the singularity at $r = 0$. We show that on those hypersurfaces sufficiently close to the singularity, \emph{every} radial solution to Jang's equation blows up. The proof depends only on the geometry in an arbitrarily small neighborhood of the singularity, suggesting that Jang's equation is in fact detecting the singularity. We comment on possible applications to the weak cosmic censorship conjecture.
\end{abstract}

\section{Introduction}
Jang's equation, which first appeared in \cite{J78}, is a valuable tool in the study of apparent horizons, on account of the following well-known existence theorem of Schoen and Yau.

\begin{theorem} [\cite{SY81}] \label{SYexist}
Let $(M,g,h)$ be an asymptotically flat initial data set. Then there exists an unbounded open set $\Omega_0 \subset M$, whose boundary is a finite (possibly empty) disjoint union of closed, smoothly embedded apparent horizons, and a function $f \in C^{2,\alpha}(\Omega_0)$, such that $f$ satisfies Jang's equation on $\Omega_0$, $|f(x)| \rightarrow \infty$ uniformly as $x \rightarrow \partial \Omega_0$, and $f(x) \rightarrow 0$ as $|x| \rightarrow \infty$ on each end of $M$.
\end{theorem}

As an immediate corollary, if $(M,g,h)$ does not contain any apparent horizons, then there exists at least one global solution to Jang's equation that decays to zero at infinity. This fact has been used to prove the existence of apparent horizons by imposing geometric conditions on $M$ that ensure \emph{no} global solutions exist \cite{SY83,E95,Y01}. Note that Theorem \ref{SYexist} only yields the existence of one global solution\,---\,it does not imply that \emph{every} solution is global in the absence of apparent horizons. More generally, Schoen and Yau's analysis shows that any solution arising from a limit of suitably regularized boundary-value problems can only blow up on an apparent horizon (see Section 3.5 of \cite{AEM10}). This was used in \cite{AM09} and \cite{Eichmair} to prove the existence of apparent horizons in the presence of suitable geometric barriers. (A simple example is when there exists a bounded set $\Omega \subset M$ with at least two boundary components and $H_{\partial \Omega} - |\text{tr}_{\,\partial \Omega}\,h| > 0$.)

The existence arguments cited above are all nonlocal, in the sense that they depend on geometric invariants and boundary properties of some region $\Omega \subset M$. In this paper we tackle the question of whether or not one can force solutions to Jang's equation to blow up as a result of purely local phenomena, such as a nearby spacetime singularity. To that end, we consider a family of spherically symmetric, spacelike Cauchy hypersurfaces $\{L_c\}$ (defined in Section \ref{JangLc}) in the maximally extended Schwarzschild spacetime. These are well-defined for all $c$ in an open interval, $c \in (-\crit,\crit)$, and have the property that the radial coordinate $r$ satisfies $\inf_{L_c} r \rightarrow 0$ as $|c| \rightarrow \crit$. Thus the hypersurfaces are getting progressively ``closer" to the singularity at $r=0$ as $c$ increases toward $\crit$. Our main theorem provides compelling evidence that Jang's equation is in fact able to detect the singularity.

\begin{theorem} \label{blowup}
There exists $c_0 < \crit$ such that if $c_0 < c < \crit$, then Jang's equation admits no global, spherically symmetric solutions on the spacelike Cauchy hypersurface $L_c$.
\end{theorem}

In the proof we construct solutions that blow up \emph{inside} the apparent horizon when $c$ is sufficiently close to $\crit$. Here we define a blowup solution in the traditional ODE sense, meaning that the maximal domain of the solution $f$ is a proper subset of $\mathbb{R}$ and so either $|f|$ or $|f'|$ is unbounded. This is different from the usual notion of a blowup for Jang's equation, which has $|f| \ra \infty$ and necessarily occurs along an apparent horizon (cf. \cite{EM12}). This point is clarified in Section \ref{asymptotic}, where we prove that $|f|$ remains bounded for the particular solutions we construct.

Theorem \ref{blowup} does not rule out the possibility of nonradial global solutions on $L_c$. However, we can show that the solution arising from the limiting construction in the proof of Theorem \ref{SYexist} (see \cite{AEM10} for details) must blow up once $c$ is sufficiently close to $\crit$. This distinction is important because, as was observed in \cite{MalecMurchadha} (and will be seen in the proof of Theorem \ref{blowup}), solutions can blow up on surfaces that are not apparent horizons.

\begin{corollary} \label{SYblowup}
Let $\hat{f}_c$ denote the solution to Jang's equation constructed in the proof of Theorem \ref{SYexist} (with $M = L_c$). If $c_0 < c < \crit$, then $\hat{f}_c$ blows up.
\end{corollary}

It is then a consequence of Schoen and Yau's blowup analysis that $L_c$ contains an apparent horizon when $c > c_0$. This conclusion is trivial in the present setting, since it will follow immediately from the definition below that $L_c$ contains an apparent horizon for \emph{any} value of $c$. However, this example suggests that Jang's equation may be important in studying the relation between singularities and apparent horizons in a broader setting.

This relation was first put forth by Penrose in \cite{P65}, the first of the famous ``singularity theorems."  A recent variant of this theorem in terms of apparent horizons (rather than the trapped surfaces in Penrose's original formulation) can be found in \cite{eichmair2012topological}.  Given the role that apparent horizons play in the weak cosmic censorship conjecture (see \cite{W84}), and the link that Theorem~\ref{SYexist} provides between Jang's equation and apparent horizons in initial data sets, our results suggest that Jang's equation is an important tool in further understanding the censorship conjecture.

Such an understanding would of course require the study of non-spherically symmetric geometries, which is beyond the scope of this paper.  An extension of Theorem~\ref{blowup} to generic spacelike hypersurfaces is expected to be highly nontrivial\,---\,for example, it is not known if similar blowup behavior occurs on the family of highly non-symmetric Cauchy hypersurfaces constructed in \cite{WI91}, which approach the singularity at $r=0$ but contain no trapped surfaces, even in their causal past. It seems unlikely that a direct generalization of Theorem~\ref{blowup} will hold for every possible family of hypersurfaces near a given spacetime singularity, but we hasten to point out that such generality is not required to deduce the existence of an apparent horizon. To wit, it suffices to construct a single family of hypersurfaces, approaching the singularity, on which Jang's equation blows up once some critical threshold has been crossed. Therefore, while the slices arising in the general case will not have the high degree of symmetry that is essential to the proof of Theorem \ref{blowup}, it may still be possible to construct a distinguished family of hypersurfaces that is amenable to direct analysis. However, without further analysis it is not possible to know the extent to which the ideas of this paper will generalize to more complicated geometries. 

A possible generalization of our approach is as follows.  Let $M$ be a singular spacetime with a hypersurface-orthogonal congruence of timelike geodesics $\{\gamma\}$, and suppose that one member of this congruence, say $\gamma_0$, is future incomplete. Let $(a,b)$ be the maximal domain of $\gamma_0(t)$, which necessarily has $b < \infty$, and let $\{\Sigma_t\}$ be a family of spacelike hypersurfaces intersecting $\gamma_0$ at $\gamma_0(t)$. If each $\Sigma_t$ is asymptotically flat (or has suitable boundary behavior), one could then perform a blowup analysis of Schoen and Yau's solution of Jang's equation (i.e., the solution whose existence is guaranteed by Theorem \ref{SYexist}) in the $t \ra b$ limit, in an attempt to conclude the existence of an apparent horizon in $\Sigma_t$.

Given the wide variety of singular behavior that may be present in a spacetime (see Chapter 9 of \cite{W84} for examples), one should not expect this approach to work universally. For instance, it seems unlikely that arguments like those in our proof of Theorem \ref{blowup} will work in the presence of a conical singularity, where the metric and its derivatives are uniformly bounded on the complement of a single point, but cannot be smoothly extended to the entire manifold \cite{T94}. The method presented in this paper seems better adapted to the case of curvature singularities. In this case we might take motivation from the Weyl Curvature Hypothesis \cite{P79,T87,G91}, which posits that the Weyl tensor becomes unbounded as one approaches a final singularity. This is certainly the case for the (maximally extended) Schwarzschild spacetime: as a vacuum solution of Einstein's equations it is necessarily Ricci flat, hence the Weyl and Riemann tensors coincide, and it is well known that the latter blows up at $r=0$.

\section{Definitions and notation}
In this section we introduce Jang's equation, then describe the maximally extended Schwarzschild solution (using Kruskal coordinates) and explain the importance of these coordinates to our analysis. Finally, we define the family of spacelike Cauchy hypersurfaces $\{L_c\}$ alluded to in the introduction, and derive Jang's equation for radial functions on $L_c$.

\subsection{Jang's equation}
We follow as much as possible the notation and conventions of the excellent recent review \cite{AEM10}. Consider an initial data set $(M,g,h)$, where $M$ is a smooth 3-manifold with Riemannian metric $g$, and $h$ is a symmetric $(0,2)$-tensor. The initial data set can be viewed intrinsically, though in practice $M$ arises as a spacelike slice in a Lorentzian 4-manifold, with $g$ and $h$ the induced metric and second fundamental form, respectively. It is of interest to know when a given initial data set corresponds to a spacelike slice of Minkowski spacetime. In \cite{J78} it was shown that this is the case precisely when there exists a smooth function $f$ such that
\beqa
	\mathfrak{h}\ :=\ h - \frac{\nabla^2 f}{\sqrt{1 + | df |^2}}\nonumber
\eeqa
vanishes and the metric $\mathfrak{g} := g + df \otimes df$ is flat. Jang's equation arises from contracting $\mathfrak{h}$ and $\mathfrak{g}$;  in local coordinates this reads
\beqa
\label{eqn:jang1}
	\left(g^{ij} - \frac{\nabla^i f \nabla^j f}{1+| df|^2}\right)\left(h_{ij} -\frac{\nabla_i \nabla_j f}{\sqrt{1+| df|^2}}\right)\ =\ 0.
\eeqa
This can be recast in a more familiar geometric form by noticing that $\mathfrak{g}$ is precisely the induced metric on the graph of $f$ in the product $(M \times \mathbb{R},g+dt^2)$. Denoting this graph by $\Gamma(f)$, we can write Jang's equation as
\beqa
	H_{\,\Gamma(f)}\ =\ \textrm{tr}_{\,\Gamma(f)}\, h,\nonumber
\eeqa
where $H_{\,\Gamma(f)}$ denotes the mean curvature of $\Gamma(f)$ computed with respect to the downward-pointing unit normal
\beqa
	\nu_{\,\Gamma(f)}\ :=\ \frac{(\nabla f, -1)}{ \sqrt{1+|df|^2}} \nonumber
\eeqa
and $\textrm{tr}_{\,\Gamma(f)}\, h$ denotes the trace of $h$ with respect to the induced metric $\mathfrak{g}$ on the graph, where $h$ is extended to $M \times \mathbb{R}$ by zero in the vertical direction.

A significant application of Jang's equation appeared in the proof of the positive mass theorem for general initial data sets \cite{SY81}, where it was used to reduce the problem to the time-symmetric case, which had already been established.

\subsection{The maximally extended Schwarzschild spacetime}
Following the notation in Chapter 13 of \cite{ON83}, we express the maximally extended Schwarzschild spacetime as a warped product $Q \times_r \mathbb{S}^2$, where $Q \subset \RR^2$ is the Kruskal plane
\beqa
	Q\ =\ \{(u,v)\,:\,uv > -2m/e\}.\nonumber
\eeqa
The warping factor $r$ is given implicitly as a function of $u$ and $v$ by the equation $\phi(r) = uv$, where $\phi(r) := e^{r/2m-1}\left(r-2m\right)$. (To see that this is well-defined, it suffices to note that $\phi_r(r) > 0$ for $r>0$, hence $\phi:(0,\infty)\lra (-2m/e,\infty)$ is a diffeomorphism.) We refer to $(u,v,\theta,\varphi)$ as the Kruskal coordinates. With respect to these coordinates, the maximally extended Schwarzschild metric of mass $m$ is given by
\beqa
\label{oneill}
\tilde{g}\ =\ F(r)\,du \otimes dv + F(r)\,dv \otimes du + r^2 d\vtheta\otimes
d\vtheta
+r^2\sin^2\vtheta\,d\vphi\otimes d\vphi,
\eeqa
where $F(r) := (8m^2/r)\,e^{1-r/2m}$. Here the $r=2m$ apparent horizon\,---\,at which the metric is singular in the Schwarzschild coordinates\,---\,corresponds to the coordinate axes $\{u=0\} \cup \{v=0\}$, where $\tilde{g}$ is easily seen to be smooth.

It is this property of the Kruskal coordinates that we desire for our study. Given the crucial geometric role the apparent horizon plays in the analysis of Jang's equation (as described in the Introduction), it is clearly inappropriate to use a coordinate system in which the metric becomes singular precisely at the point of interest. Kruskal coordinates circumvent this problem, allowing us to simultaneously study solutions on either side of the horizon.

\subsection{The $L_c$ hypersurfaces}
\label{JangLc}
The family of Cauchy hypersurfaces $L_c$ (which appeared in the statement of Theorem \ref{blowup} in the Introduction) are defined by $L_c := \{(u,v,\vtheta,\vphi)\,:\,v = u+c\}$ (see Figure \ref{fig1}; note that $L_0$ is the familiar $t=0$ Cauchy time slice in Schwarzschild coordinates). The condition $uv > -2m/e$ implies that $L_c$ is well-defined as long as
\beqa
	|c| < \crit\ :=\ \sqrt{\frac{8m}{e}}\cdot
\eeqa
It is shown in the Appendix that on each $L_c$ Jang's equation is given by
\beqa
\label{jang}
f''\ =\ 
c\sqrt{\frac{2m}{\phi_r(r)}+\frac{(f')^2}{4}}\hspace{-.12in}&&\hspace{-.12in}\left[\frac{1}{\phi_r(r)}\left(\frac{1}{2m}
-\frac{3}{r}\right) - \frac{(f')^2}{2mr}\right]\\
&& -\ \frac{2u+c}{4mr}(f')^3\ -\ \frac{2u+c}{2\phi_r(r)}\left(\frac{1}{2m} + \frac{5}{r}\right)f'\,\nonumber
\eeqa
for any function $f(u)$, where primes denote differentiation with respect to $u$. Since $r$ is implicitly a function of $u$ and $v$, we can write $r = r(u)$ for the unique radial coordinate of
the point $(u,v(u)) = (u,u+c)$ on $L_c$.

%%%%%%%%%%figure1
\begin{figure}[t!]
\begin{center}
\includegraphics[scale=.8]{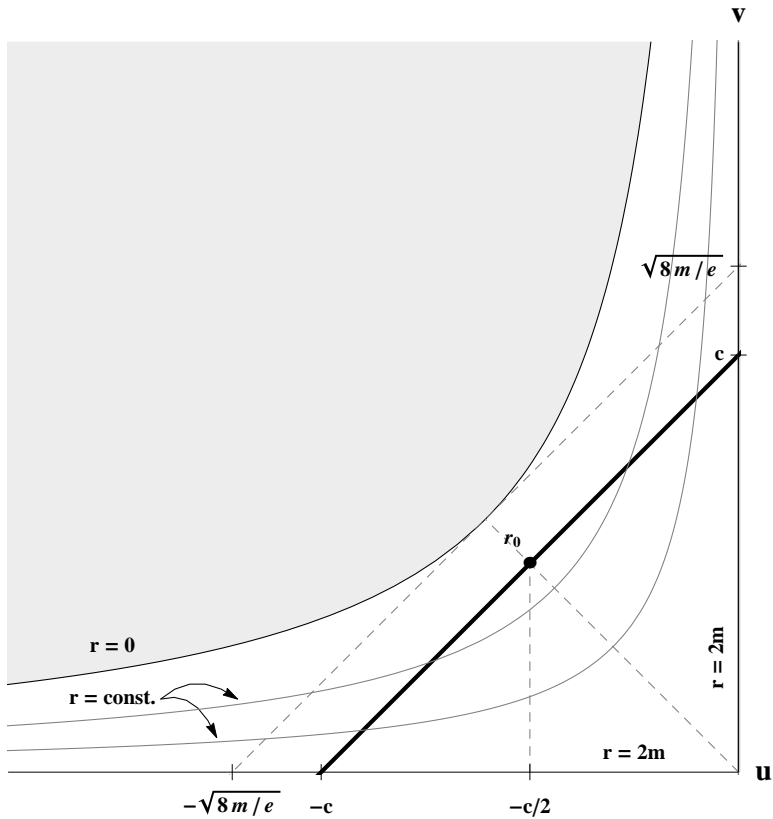}
\end{center}
\caption{The interior black hole region $\{u < 0\} \cap \{v > 0\}$ in the Kruskal plane $Q$. Note that each point in $Q$ corresponds to a 2-sphere in the maximally extended Schwarzschild spacetime. The apparent horizon at $r=2m$ corresponds to the coordinates axes $\{u=0\} \cup \{v=0\}$.
The solid black line is the spherically symmetric, spacelike Cauchy hypersurface $L_c$  defined by $v(u) = u+c$. The minimum value of $r$ on $L_c$ is denoted $r_0$; this value is achieved at $u=-c/2$ and is strictly positive provided $c< \sqrt{8m/e}$. Our main theorem describes the behavior of Jang's equation on $L_c$ as $c \ra \sqrt{8m/e}$ (hence $r_0 \rightarrow 0$).}
\label{fig1}
\end{figure}

\section{Blowup solutions of Jang's equation on $L_c$}
To produce a blowup solution, we consider Jang's equation \eqref{jang} with the initial conditions
\beqa
\label{initialcond}
f(-c/2)\ =\ f'(-c/2)\ =\ 0.
\eeqa
By the Picard--Lindel\"{o}f theorem, there exists a unique, smooth solution on a maximal interval $I_c \subset \RR$ that contains $u=-c/2$. Recall from Figure \ref{fig1} that $(u,v)=(-c/2,c/2)$ is precisely the point in $L_c$ at which the minimum value of $r$ is achieved, hence we are prescribing initial data at the point ``closest" to the singularity at $r=0$. The condition $f_c(-c/2) = 0$ is imposed merely for convenience, since if $f_c(u)$ solves \eqref{jang} then so does $f_c(u) + \beta$ for any $\beta \in \RR$. The condition $f_c'(-u/2)=0$ ensures that $f_c$ is symmetric about the line $\{v=-u\}$. This symmetry, demonstrated in Lemma \ref{lemma:2h} below, will be crucial to our blowup construction.

The case $c=0$ is trivial, with $f_0(u) \equiv 0$ for all $u$, and $L_{-c}$ is equivalent to $L_c$, so it suffices to consider $0 < c < \crit$. We can thus define a new variable
\beqa
\label{us}
s\ =\ \gamma(u)\ :=\ \frac{2}{c}u + 1,
\eeqa
on $\RR$, which in particular maps the interval $[-c,0]$ to $[-1,1]$.  Next, we define
\beqa
\label{hdef}
h(s)\ :=\ \frac{f'(\gamma^{-1}(s))}{\sqrt{1+f'(\gamma^{-1}(s))^2}}
\eeqa
and observe that
\beqa
\label{hinv}
f'(\gamma^{-1}(s))\ =\ \frac{h(s)}{\sqrt{1-h(s)^2}},
\eeqa
so $|f'(\gamma^{-1}(s))|$ blows up precisely when $|h(s)| = 1$. (Note that $'$ always denotes differentiation with respect to $u$, and never $s$.) Rewriting \eqref{jang} in terms of $h$, we obtain
\beqa
\label{jang2}
\frac{dh}{ds}\ =\ \frac{c^2}{2}\hspace{-.125in}&&\hspace{-.125in}\sqrt{\frac{2m(1-h^2)}{\phi_r(r)} + \frac{h^2}{4}}\left[\frac{(1-h^2)}{\phi_r(r)}\left(\frac{1}{2m} - \frac{3}{r}\right) - \frac{h^2}{2mr}\right]\\
&&\hspace{.85in}-\,\frac{c^2 s}{8mr}\,h^3 - \frac{c^2 s}{4\phi_r(r)}\left(\frac{1}{2m}+\frac{5}{r}\right) h(1-h^2).\nonumber 
\eeqa
Henceforth we call \eqref{jang2} the \textit{transformed Jang equation}.  The function $h$ is more convenient to work with than $f$, and is a frequently utilized quantity in the study of Jang's equation with spherical symmetry \cite{MalecMurchadha,W10}.  

\begin{lemma}
\label{lemma:jangh}
For any $c \in (0,\crit)$, if $f_{c}\colon I_c \lra \RR$ denotes the maximal solution to Jang's equation~\eqref{jang} with the initial conditions \eqref{initialcond}, then the function $h_{c}\colon \gamma(I_c) \lra \RR$ defined in \eqref{hdef} is a solution to the transformed Jang equation \eqref{jang2}, with $h_c(0) = 0$ and $|h_c(s)| < 1$ for all $s \in \gamma(I_c)$.

Conversely, if $h_c: J_c \lra \RR$ denotes the maximal solution to the transformed Jang equation with initial condition $h_c(0) = 0$ and $J_c^1 := \{s \in J_c : |h(s)| < 1\}$, then the function $f_c: \gamma^{-1}(J_c^1) \lra \RR$ given by \eqref{hinv} (with $f_c(-c/2) = 0$) is the maximal solution to \eqref{jang} with the initial conditions \eqref{initialcond}.
\end{lemma}

\begin{proof}
The correspondence between $f_c$ and $h_c$ follows immediately from the relation between the two equations \eqref{jang} and \eqref{jang2}, which is the result of an elementary computation. It only remains to prove the maximality of $f_c$ asserted in the second half of the lemma. Suppose this is not true, so there exists a function $\tilde{f}_c$ satisfying \eqref{jang} and \eqref{initialcond} on an open interval $\tilde{I}_c$ that strictly contains $\gamma^{-1}(J_c^1)$. Then the corresponding function $\tilde{h}_c$, defined by \eqref{hdef}, satisfies the transformed Jang equation on $\gamma(\tilde{I}_c)$, with $\tilde{h}_c(0) = 0$, and so $\tilde{h}_c$ agrees with the maximal solution $h_c$ on $\gamma(\tilde{I}_c)$. But this is not possible, because $|\tilde{h}_c(s)| < 1$ on $\gamma(\tilde{I}_c)$, whereas $|h_c(s)| \geq 1$ on $\gamma(\tilde{I}_c) \setminus J_c^1$.
\end{proof}

In particular, if $|h_c(s)| = 1$ for some $s$, then the maximal solution $f_c$ to \eqref{jang} and \eqref{initialcond} cannot be defined for all $u \in \RR$.

\begin{lemma}
\label{lemma:2h}
If $h(s)$ is a solution to the transformed Jang equation \eqref{jang2}, then $-h(-s)$ is also a solution. In particular, the unique solution $h_c$ with $h_c(0) = 0$ is an odd function of $s$.
\end{lemma}

\begin{proof}
We write \eqref{jang2} in the general form 
\beqa
\label{eqn:F}
\frac{dh}{ds}\ =\ F_c(s,h),
\eeqa
recalling that $r$ is defined implicitly as a function of $s$. Using \eqref{us}, we have
\beqa
\label{eqn:rs}
	\phi(r)\ =\ u(u+c)\ =\ \frac{c^2}{4} (s^2-1),
\eeqa
and so $r(s) = r(-s)$.  It follows immediately that $F_c(s,h) = F_c(-s,-h)$.  Now let $h(s)$ be a solution to \eqref{eqn:F} and consider the function $g(s) := -h(-s)$.  By the chain rule,
$$
\frac{dg}{ds}\ =\ F_c(s,g(s)),
$$
hence $g$ also solves \eqref{eqn:F}.
\end{proof}

Therefore, if we can prove existence of an $s_0 > 0$ such that $h_c(s_0) = -1$, it will follow immediately that $h_c(-s_0) = 1$, hence the range of $h_c$ contains the closed interval $[-1,1]$. Then if $f$ satisfies Jang's equation (with \emph{any} initial values) on $[-\gamma^{-1}(s_0),\gamma^{-1}(s_0)]$, the graph of the corresponding $h$ must intersect the graph of $h_c$, which is only possible if the two functions coincide.

To establish the existence of such an $s_0$, we require a $c$-independent bound on $dh_{c}/ds$. The first task in deriving such an estimate is understanding the rate at which the minimum value of $r$ (which is achieved at $s=0$) approaches $0$ as $c \ra \crit$.

\begin{lemma} \label{rbound}
Let $c \in (0,\crit)$, and define
\beqa
	\delta\ :=\ \sqrt{em(\crit^2-c^2)}.
	\label{deltadef}\nonumber
\eeqa
Then
\beqa
	r(s)\ \leq\ 2\sqrt{2}m s + \delta\nonumber
\eeqa
for all $s \in \RR$.
\end{lemma}

\begin{proof}
The function $\phi(r)$ satisfies $\phi_r(r) \geq r/2em$ and $\phi(0) = -2m/e$, so we integrate to find that
\begin{align*}
	\phi(r)\ \geq\ -\frac{2m}{e}\ +\ \frac{r^2}{4em}\cdot
\end{align*}
Using \eqref{eqn:rs}, this implies
\begin{align*}
	r^2\ \leq\ emc^2 s^2\ +\ (8m^2 - emc^2).
\end{align*}
Observing that $emc^2 < 8m^2$ because $c < \sqrt{8m/e}$, the result now follows from the basic inequality $\sqrt{a+b} \leq \sqrt{a} + \sqrt{b}$.
\end{proof}

We are now ready to prove the main estimate used in our blowup construction, bearing in mind that $F_c(s,h)$ denotes the right-hand side of the transformed Jang equation (as in \eqref{eqn:F}).
 
\begin{lemma}
\label{lemma:Fbound}
Fix $\bar{c} \in (0,\crit)$. There exist positive constants $A$ and $B$ such that, if $c \geq \bar{c}$, then
\beqa
	F_c(s,h)\ \leq\ -B\,\frac{(1-h^2)^2 - \,h^3}{As + \delta}
\eeqa
for all $(s,h) \in [0,1/3\sqrt{2}] \times [-1,0]$.
\end{lemma}

The key point is that the above estimate is valid for any $c \geq \bar{c}$, so it holds uniformly as $c \ra \crit$.

\begin{proof}
We first restrict our attention to the interval $s \in [0,1]$, so we have $r \leq 2m$, $h \leq 0$, and $1-h^2 \geq 0$. Using the geometric inequality $\sqrt{2(a+b)} \geq \sqrt{a} + \sqrt{b}$, we find
\beqa
\frac{c^2}{2}\underbrace{\sqrt{\frac{2m(1-h^2)}{\phi_r(r)} + \frac{h^2}{4}}}_{\geq\,\sqrt{\frac{m(1-h^2)}{\phi_r(r)}}\, -\, \frac{h}{2\sqrt{2}}}\,\underbrace{\frac{(1-h^2)}{\phi_r(r)}\left(\frac{1}{2m}- \frac{3}{r}\right)}_{<\,0} &\leq&\nonumber \\
\frac{c^2\sqrt{m}}{2\phi_r(r)^{3/2}}\left(\frac{1}{2m} - \frac{3}{r}\right)(1-h^2)^{3/2}\
-\ \frac{c^2}{4\sqrt{2}\,\phi_r(r)}\left(\frac{1}{2m}- \frac{3}{r}\right)\hspace{-.175in}&&\hspace{-.175in}h\,(1-h^2)\nonumber
\eeqa
and
\beqa
-\frac{c^2}{2}\underbrace{\sqrt{\frac{2m(1-h^2)}{\phi_r(r)} + \frac{h^2}{4}}}_{\geq\,\frac{-h}{2}}\ \underbrace{\,\frac{h^2}{2mr}\,}_{>\,0} &\leq& \frac{c^2}{8mr}\,h^3.\nonumber
\eeqa
Hence on the interval $[0,1]$,
\beqa
\label{eqn:ineq1}
F_c(s,h) &\leq& \frac{c^2(1-s)}{8mr}\,h^3\ +\ \frac{c^2\sqrt{m}}{2\phi_r(r)^{3/2}}\left(\frac{1}{2m} - \frac{3}{r}\right)(1-h^2)^{3/2}\\
&&\hspace{.5in} +\ \frac{c^2}{4\phi_r(r)}\left(\frac{3-5\sqrt{2}s}{\sqrt{2}r} - \frac{1+\sqrt{2}s}{2\sqrt{2}m}\right)h\,(1-h^2).\nonumber
\eeqa
Now further restrict to the subinterval $[0,1/(3\sqrt{2})]$.  Since $r < 2m$, the third term on the right-hand side of \eqref{eqn:ineq1} is bounded above on this subinterval by
$$
\frac{c^2}{4\phi_r(r)}\left(\frac{1-3\sqrt{2}s}{\sqrt{2}m}\right)h\,(1-h^2)\ \leq \ 0,
$$
hence can be discarded. Using the hypothesis $c \geq \bar{c}$, the first term on the right-hand side of \eqref{eqn:ineq1} is seen to satisfy
$$
\frac{c^2(1-s)}{8mr}\,h^3\ \leq\ \frac{\bar{c}^2(1-1/3\sqrt{2})}{8mr}\,h^3.
$$
Finally we observe that $\phi_r(r) = (r/2m) e^{r/2m-1} < 1$ and $(1-h^2)^{3/2} \geq (1-h^2)^2$, so the second term on the right-hand side of \eqref{eqn:ineq1} is bounded above by
$$
-\frac{\bar{c}^2\sqrt{m}}{r}\,(1-h^2)^2.
$$

Setting $A := 2\sqrt{2}m$ and applying Lemma \ref{rbound}, we have
$$
	F_c(s,h) \ \leq\ -\bar{c}^2\sqrt{m} \Bigg[\frac{(1-h^2)^2}{As+\delta}\ -\ \frac{(1-1/3\sqrt{2} )}{8m^{3/2}}\,\frac{h^3}{As+\delta}\Bigg]\hspace{.2in}
$$
for $s \in [0,1/3\sqrt{2}]$. The result follows, with
\beqa
	B\ :=\ \bar{c}^2 \sqrt{m} \min \left\{ 1, \frac{(1-1/3\sqrt{2} )}{8m^{3/2}} \right\}\cdot
\eeqa
\end{proof}

We are now ready to prove our main result.

\begin{proof}[Proof of Theorem \ref{blowup}] We first show that $f_c$ blows up when $c$ is close enough to $\crit$. By Lemma \ref{lemma:jangh}, it suffices to find $s >0$ with $|h_c(s)| = 1$. Assume this is not the case, so $|h_c(s)| < 1$ for all $s$. It is easy to see from \eqref{jang2} that $F(s,0) < 0$ when $r \leq 2m$; this implies $h_c(s) \leq 0$ for $s \in [0,1]$ (and in fact $h_c(s) < 0$ for $s \in (0,1]$, though this stronger version will not be needed).

Therefore we have from Lemma \ref{lemma:Fbound} that
\beqa
\frac{dh_c}{ds}\ \leq\ -B\,\frac{(1-h_c^2)^2 - \,h_c^3}{As + \delta}\nonumber
\eeqa
for $s \in [0,1/3\sqrt{2}]$. This inequality is separable, so letting $\psi$ denote an antiderivative of the function $h \mapsto [(1-h^2)^2 - h^3]^{-1}$, we find that
$$
\psi(h_c(s))\ -\ \psi(h_c(0))\ \leq\ \frac{B}{A}\,\log\left( \frac{\delta}{As+\delta} \right)\cdot
$$
On the other hand, there is a constant $C > 0$ such that $(1-h^2)^2 - h^3 \geq C > 0$ for $h \in [-1,0]$, hence
$$
\psi(0)\ -\ \psi(h)\ =\ \int_h^0 \frac{dh}{(1-h^2)^2 - h^3}\ \leq\ -\frac{h}{C}\cdot
$$
Combining the above two inequalities, we find that
$$
	h_c(s)\ \leq\ \frac{BC}{A}\,\log\left( \frac{\delta}{As+\delta} \right),
$$
with $A,B$ and $C$ independent of $c$. Since $\delta \ra 0$ as $c \ra \crit$
and
$$
	\lim_{\delta \ra 0} \,\log\left( \frac{\delta}{As_0+\delta} \right)\ =\ -\infty
$$
for any fixed $s_0 > 0$, we conclude that there exists $c_0 > 0$ such that, if $c > c_0$, then $h_c(s_*) = -1$ for some $s_* \in [0,1]$. This completes the proof that for every $c$ sufficiently close to $\crit$, $f_c$ must blow up.

Now assume $c > c_0$, and let $f$ denote any solution to Jang's equation on $L_c$, with $h$ the corresponding solution to the transformed Jang equation defined by \eqref{hdef}. Suppose $f$ is a global solution, so $|h(s)| < 1$ for all $s \in \RR$. We know from Lemma \ref{lemma:2h} that $h_c$ is an odd function, hence its range contains the entire closed interval $[-1,1]$. Therefore the graphs of $h$ and $h_c$ must intersect, which is only possible if they coincide. This contradicts the bound $|h(s)| < 1$ and completes the proof.
\end{proof}

We conclude with the proof of Corollary \ref{SYblowup}. By Theorem \ref{blowup}, it suffices to show that the solution $\hat{f}_c$ constructed by Schoen and Yau is spherically symmetric. This is an immediate consequence of the following general result.

\begin{lemma} Let $(M,g,h)$ be an asymptotically flat initial data set, and let $\hat{f}$ denote the solution to Jang's equation given by Theorem \ref{SYexist}. If $\varphi:M \rightarrow M$ is a diffeomorphism such that $\varphi^*g = g$ and $\varphi^*h = h$, then $\hat{f} = \hat{f} \circ \varphi$.
\end{lemma}

In other words, symmetries of the initial data give rise to symmetries of $\hat{f}$. It is important to note that this does not hold for arbitrary solutions to Jang's equation.

\begin{proof} Recalling the proof of Theorem \ref{SYexist}, we consider the regularized equation
\beqa
	\label{Jangreg}
	H_{\,\Gamma(f_t)} - \text{tr}_{\,\Gamma(f_t)}\,h\ =\ tf_t.
\eeqa
For each $t > 0$ there exists a unique solution $f_t$ that satisfies $f_t(x) \ra 0$ as $|x| \ra \infty$ on each end of $M$. Since \eqref{Jangreg} is coordinate invariant, we find that $f_t \circ \varphi$ is also a solution, which implies $f_t = f_t \circ \varphi$.

We now observe that $\Omega_0$, the domain of $\hat{f}$, is precisely the set where $\{f_t\}$ converges as $t \ra 0$, with
\[
	\hat{f}(x)\ :=\ \lim_{t \ra 0} f_t(x)
\]
for all $x \in \Omega_0$, and the result follows.
\end{proof}

\section{Local structure of the blowup solutions} \label{asymptotic}
Our final task is to describe the asymptotic structure of the blowup solutions $f_c$ constructed above.

\begin{theorem}
Let $f_c$ solve Jang's equation \eqref{jang} with initial conditions \eqref{initialcond}, and $h_c$ solve the transformed Jang equation \eqref{hdef} with $h_c(0) = 0$. Suppose $h_c(s_0) = -1$ for some $s_0 \in (0,1)$, and define $u_0 = \gamma^{-1}(s_0)$. Then there are positive constants $K$ and $\epsilon$ such that
\[
	f_c'(u)\ \geq\ - \frac{K}{\sqrt{u_0 - u}}
\]
for any $u \in [u_0 - \epsilon,u_0)$. In particular, 
\[
	\lim_{u \ra u_0} f_c(u)
\]
exists and is finite.
\end{theorem}

Therefore, solutions that blow up inside the horizon have the property that $f$ remains bounded while $f'$ diverges. 

\begin{proof}
By assumption we have $h_c(s_0) = -1$. Since $s_0 < 1$, this implies
\[
	\frac{dh_c}{ds}(s_0)\ =\ F_c(s_0,-1)\ =\ \frac{c^2(s_0-1)}{8mr(s_0)}\ <\ 0.
\]
It follows from Taylor's theorem that there exist positive constants $C$ and $\epsilon$ such that
\[
	h_c(s)^2 \ \leq \ 1 + C(s-s_0)
\]
when $s \in [s_0 - \epsilon, s_0]$. Therefore $1-h_c(s)^2 \geq -C(s-s_0)$, and so
\[
	|f_c'(\gamma^{-1}(s))|^2\ =\ \frac{h_c(s)^2}{1-h_c(s)^2}\ \leq\ \frac{1 + C(s-s_0)}{-C(s-s_0)}\ \leq\ \frac{1}{C(s_0-s)}\cdot
\]
Recalling that $f_c' < 0$, we use the relation $s_0-s = (2/c)(u_0-u)$ to find that
\[
	f_c'(u)\ \geq\ -\frac{K}{\sqrt{u_0-u}}
\]
with $K^{-1} := \sqrt{2C/c}$, as desired.

Finally, we integrate to obtain
\[
	f_c(u) - f_c(u_0 - \epsilon) \ \geq\ - \int_{u_0 - \epsilon}^u \frac{K}{\sqrt{u_0-v}}\,dv\ =\ 2K \left[\sqrt{u_0-\epsilon}\, - \sqrt{u} \right]
\]
for any $u \in [u_0 - \epsilon, u_0)$. Since $f_c'(u) < 0$ for $u \in (-c/2,u_0)$, we have that $f_c(u)$ is monotone decreasing and bounded below as $u \ra u_0$, hence the claimed limit exists.
\end{proof}

\section{Acknowledgments}
The authors thank Hubert Bray for suggesting this application of Jang's equation, Kuo--Wei Lee for carrying out some preliminary computations, and Michael Eichmair for numerous helpful discussions.
This work was supported by the World Premier International Research Center Initiative (WPI Initiative), MEXT, Japan.

\section{Appendix}

In this appendix we derive Jang's equation \eqref{eqn:jang1} on any hypersurface in the maximally extended Schwarzschild spacetime $(\widetilde{M},\tilde{g})$ of the form $M := \{u,\nu(u),\vtheta,\vphi\}$, where $\nu(u)$ is any smooth function of $u$ satisfying $\nu'(u) > 0$.  The hypersurfaces $L_c$ considered above then correspond to $\nu(u) = u+c$. Defining the smooth function
\beqa
\label{level}
\rho\colon \widetilde{M}\lra \RR\hspace{.5in},\hspace{.5in}(u,v,\vtheta,\vphi)\ \mapsto\ \nu(u) - v,\nonumber
\eeqa
it is easy to see that the differential $d\rho\colon T_{(u,\nu(u),\vtheta,\vphi)} \lra T_{0}\RR$ is surjective.  Hence $\rho^{-1}(0) = M$ is a (closed) hypersurface. Now we compute
the (scalar) second fundamental form 
of $M$.  We will work in the coordinates provided by the smooth embedding
$\iota\colon(u,\vtheta,\vphi)
\inc (u,\nu(u),\vtheta,\vphi)$, in terms of which the induced metric on $M$ is
\beqa
\tilde{g}|_{M} &:=& g\nonumber\\
&=& 2F(r)\nu'(u)\,du \otimes du + r^2 d\vtheta \otimes d\vtheta
+r^2\sin^2\vtheta\,d\vphi \otimes d\vphi.\label{inducedg2}
\eeqa
Note that $\nu'(u) > 0$ ensures that $M$ is spacelike.  Since $\iota$ is an embedding, its differential $d\iota$ is injective at each point; hence $\{
d\iota(\partial_u),d\iota(\partial_\vtheta),d\iota(\partial_\vphi)\}$ will be a
coordinate basis for $M$ in $\widetilde{M}$:
\beqa
\label{basis}
d\iota(\partial_u) & = &\partial_u + \nu'(u)\partial_v,\nonumber\\
d\iota(\partial_\vtheta) & = &\partial_\vtheta,\\
d\iota(\partial_\vphi) & = &\partial_\vphi\nonumber
\eeqa
(we use the same indices for both coordinates).  To find a unit vector normal to
$M$, we note that because $M = \rho^{-1}(0)$,
$\pm{\rm grad}\,\rho$ must be normal to $M$:
%In coordinates these are
\beqa
\pm{\rm grad}\,\rho &=& \pm\sum_{i,j}\left(\tilde{g}^{ij}\partial_j \rho\right) \partial_i\nonumber\\
&=& \pm(\tilde{g}^{uv} \nu'\partial_v - \tilde{g}^{uv}\partial_u)\nonumber\\
&=& \pm F^{-1}(\nu'\partial_v - \partial_u).\nonumber
\eeqa
Since $\ip{\pm{\rm grad}\,\rho}{\pm{\rm grad}\,\rho}{\tilde{g}} = -2F^{-1}\nu' < 0$ ($\pm{\rm
grad}\,\rho$
is timelike, as it must be), the future-pointing unit normal vector field (with respect to the time orientation $\partial_v - \partial_u$ of $\widetilde{M}$) is
$$
%\beqa
%\label{unit}
N\ :=\ \frac{{\rm grad}\,\rho}{\sqrt{2\nu'F^{-1}}}\ =\ \frac{\nu'\partial_v - \partial_u}{\sqrt{2\nu'F}} \cdot\nonumber
%\eeqa
$$
%Note that $N$ is also future-pointing because 
%$\ip{N}{\partial_v-\partial_u}{\tilde{g}} = -\frac{1+\nu'}{\sqrt{2\nu'F}}\tilde{g}_{uv} < 0$, though that's not important here (recall Appendix~\ref{MOTS}).
%; having said that, note that $N$ points {\it downward} along the $v$-direction. 
%; see \cite[p.~391]{ON83}.
%The covector field
%metrically equivalent to $N
%$ is
%$$
%\ip{N}{\cdot}{\tilde{g}} = \frac{\ip{\partial_u}{\cdot}{\tilde{g}} -
%\nu'\ip{\partial_v}{\cdot}{\tilde{g}}}{\sqrt{2\nu'F}} = 
%\frac{\tilde{g}_{uv}dv - \nu'\tilde{g}_{uv}du}{\sqrt{2\nu'F}} = \frac{-\nu'du +
%dv}{\sqrt{2\nu'F^{-1}}},
%$$
%so that
%\beqa
%\label{covN}
%N_i = (N_u,N_v,N_\vtheta,N_\vphi) =
%\left(-\sqrt{\frac{\nu'F}{2}},\sqrt{\frac{F}{2\nu'}},0,0\right).
%\eeqa
The (scalar) second fundamental form
$h$ of $M$ with respect to $N$ is defined to be $I\!I(X,Y) = h(X,Y)N$, with
$\ip{N}{N}{\tilde{g}} = -1$.
%, so the Weingarten equation actually gives
%$
%\ip{\widetilde{\nabla}_XN}{Y}{\tilde{g}} = -\ip{N}{I\!I(X,Y)}{\tilde{g}} =
%h(X,Y)$.
%; see
%\cite[p.~136]{L97}.
It is then easy to verify that $h = \widetilde{\nabla}\ip{N}{\cdot}{\tilde{g}}$.
%(because $\ip{N}{X}{\tilde{g}} \equiv 0$ for all smooth vector fields $X$ on $M$).
%; note also that since $M$ is closed, one can always extend $X$ to a smooth vector field on a neighborhood of $M$ in $\widetilde{M}$).
Now, since
the
basis for $M$ we are 
working with is $\{\partial_u + \nu'\partial_v,\partial_\vtheta,\partial_\vphi\}$
(see (\ref{basis})), the 
component $``h_{uu}"$ is
\beqa
%h_{uu} &=& h(\partial_u + \nu'\partial_v,\partial_u + \nu'\partial_v)\nonumber\\
h_{uu} &=& \widetilde{\nabla}\ip{N}{\cdot}{\tilde{g}}(\partial_u +
\nu'\partial_v,\partial_u
+ \nu'\partial_v)\nonumber\\
&=& (\widetilde{\nabla}_{\partial_u +
\nu'\partial_v}\,\ip{N}{\cdot}{\tilde{g}})(\partial_u + \nu'\partial_v)
\nonumber\\
&=& (\partial_u + \nu'\partial_v)\underbrace{\ip{N}{\partial_u +
\nu'\partial_v}{\tilde{g}}}_{0}\ -\ \ip{N}{\widetilde{\nabla}_{\partial_u +
\nu'\partial_v}(\partial_u +
\nu'\partial_v)}{\tilde{g}}\nonumber\\
&=& -\ip{N}{\widetilde{\nabla}_{\partial_u}(\partial_u +
\nu'\partial_v)}{\tilde{g}}\ -\ \nu'\ip{N}{
\widetilde{\nabla}_{\partial_v}(\partial_u +
\nu'\partial_v)}{\tilde{g}}\nonumber\\
&=& -\ip{N}{\!\!\underbrace{\widetilde{\Gamma}^p_{uu}}_{\widetilde{\Gamma}^u_{uu}
\neq \,0}\!\!
\partial_p + \nu''\partial_v + \nu'\underbrace{\widetilde{\Gamma}^p_{uv}}_{{\rm
all}~0}\partial_p}{\tilde{g}}\ -\ 
\nu'\ip{N}{\underbrace{\widetilde{\Gamma}^p_{vu}}_{{\rm all}~0}\partial_p
+ \nu'\!\!
\underbrace{\widetilde{\Gamma}^p_{vv}}_{\widetilde{\Gamma}^v_{vv} \neq
\,0}\!\!\partial_p}{\tilde{g}}\nonumber\\
&=&
-\underbrace{\widetilde{\Gamma}^u_{uu}}_{F^{-1}F_{u}}\!\ip{N}{
\partial_u}{\tilde{g}}\ -\ 
\nu''\ip{N}{\partial_v}{\tilde{g}}\ -\
(\nu')^2\!\!\!\underbrace{\widetilde{\Gamma}^v_{vv}}_{F^{-1}F_{v}}\!\!
\ip{N}{\partial_v}{\tilde{g}}\nonumber\\
%&=& \frac{\nu'F_{u}}{\sqrt{2\nu'F}} - \frac{\nu''F}{\sqrt{2\nu'F}} -
%\frac{(\nu')^2F_{v}}{\sqrt{2\nu'F}}\nonumber\\
&=& \frac{1}{\sqrt{2\nu'F}}\left[-\nu'F_{u} + \nu{''}F + (\nu')^2F_{v}\right]\nonumber\\
&=& \frac{1}{\sqrt{2\nu'F}}\left[\nu{''}F + (\nu')^2F_{r}\,r_v - \nu'F_{r}\,r_u\right].\nonumber
\eeqa
($F_u = \partial F/\partial u, F_v = \partial F/\partial v$, $F_r = dF/dr$.)  The other diagonal entries are
\beqa
h_{\vtheta\vtheta} &=& \widetilde{\nabla}\ip{N}{\cdot}{\tilde{g}}(\partial_\vtheta,\partial_\vtheta)\nonumber\\
&=& (\widetilde{\nabla}_{\partial_\vtheta}\,\ip{N}{\cdot}{\tilde{g}})\,\partial_\vtheta
\nonumber\\
&=& \partial_\vtheta\underbrace{\ip{N}{\partial_\vtheta}{\tilde{g}}}_{0}\ -\ \ip{N}{\widetilde{\nabla}_{\partial_\vtheta}\partial_\vtheta}{\tilde{g}}\nonumber\\
&=& -\ip{N}{\widetilde{\Gamma}^p_{\vtheta\vtheta}\partial_p}{\tilde{g}}\nonumber\\
&=&\frac{1}{\sqrt{2\nu'F}}\!\!\!\!\underbrace{\widetilde{\Gamma}^v_{\vtheta\vtheta}}_{-F^{-1}r\,r_{u}}\!\!\!\tilde{g}_{uv}\ -\ \frac{1}{\sqrt{2\nu'F}}\,\nu'\!\!\!\!\underbrace{\widetilde{\Gamma}^u_{\vtheta\vtheta}}_{-F^{-1}r\,r_{v}}\!\!\!\tilde{g}_{uv}\nonumber\\
&=& \frac{r}{\sqrt{2\nu'F}}\left[\nu'r_{v} - r_{u}\right],\nonumber
\eeqa
\beqa
h_{\vphi\vphi} &=& \widetilde{\nabla}\ip{N}{\cdot}{\tilde{g}}(\partial_\vphi,\partial_\vphi)\nonumber\\
&=& (\widetilde{\nabla}_{\partial_\vphi}\,\ip{N}{\cdot}{\tilde{g}})\,\partial_\vphi
\nonumber\\
&=& \partial_\vphi\underbrace{\ip{N}{\partial_\vphi}{\tilde{g}}}_{0}\ -\ \ip{N}{\widetilde{\nabla}_{\partial_\vphi}\partial_\vphi}{\tilde{g}}\nonumber\\
&=& -\ip{N}{\widetilde{\Gamma}^p_{\vphi\vphi}\partial_p}{\tilde{g}}\nonumber\\
&=&\frac{1}{\sqrt{2\nu'F}}\!\!\!\!\!\!\!\!\!\!\underbrace{\widetilde{\Gamma}^v_{\vphi\vphi}}_{-F^{-1}\sin^2\!\vtheta\,r\,r_{u}}\!\!\!\!\!\!\!\!\!\!\tilde{g}_{uv}\ -\ \frac{1}{\sqrt{2\nu'F}}\,\nu'\!\!\!\!\!\!\!\!\!\!\!\underbrace{\widetilde{\Gamma}^u_{\vphi\vphi}}_{-F^{-1}\sin^2\!\vtheta\,r\,r_{v}}\!\!\!\!\!\!\!\!\!\tilde{g}_{uv}\nonumber\\
&=& \frac{r \sin^2\vtheta}{\sqrt{2\nu'F}}\left[\nu'r_{v} - r_{u}\right],\nonumber
\eeqa
while the off-diagonal entries are all zero:
\beqa
h_{u\vtheta}\ =\ h_{\vtheta u} &=&
\widetilde{\nabla}\ip{N}{\cdot}{\tilde{g}}\,(\partial_\vtheta,\partial_u + 
\nu'\partial_v)\nonumber\\ 
&=& (\widetilde{\nabla}_{\partial_u +
\nu'\partial_v}\,\ip{N}{\cdot}{\tilde{g}})\partial_\vtheta\nonumber\\
&=& (\partial_u +
\nu'\partial_v)\underbrace{\ip{N}{\partial_{\vtheta}}{\tilde{g}}}_{0}\ -\
\ip{N}{
\underbrace{\widetilde{\nabla}_{\partial_u +
\nu'\partial_v}\partial_\vtheta}_{0}}{\tilde{g}} = 0,\nonumber\\
h_{u\vphi}\ =\ h_{\vphi u} &=&
\widetilde{\nabla}\ip{N}{\cdot}{\tilde{g}}\,(\partial_\vphi,\partial_u + 
\nu'\partial_v)\nonumber\\ 
&=& (\widetilde{\nabla}_{\partial_u +
\nu'\partial_v}\,\ip{N}{\cdot}{\tilde{g}})\partial_\vphi\nonumber\\
&=& (\partial_u +
\nu'\partial_v)\underbrace{\ip{N}{\partial_{\vphi}}{\tilde{g}}}_{0}\ -\
\ip{N}{
\underbrace{\widetilde{\nabla}_{\partial_u +
\nu'\partial_v}\partial_\vphi}_{0}}{\tilde{g}} = 0,\nonumber\\
h_{\vtheta\vphi}\ =\ h_{\vphi\vtheta} &=& \widetilde{\nabla}\ip{N}{\cdot}{\tilde{g}}\,(\partial_\vphi,\partial_\vtheta)\nonumber\\ 
&=& (\widetilde{\nabla}_{\partial_\vtheta}\,\ip{N}{\cdot}{\tilde{g}})\partial_\vphi\nonumber\\
&=& \partial_\vtheta\underbrace{\ip{N}{\partial_{\vphi}}{\tilde{g}}}_{0}\ -\
\underbrace{\ip{N}{\widetilde{\nabla}_{\partial_\vtheta}\partial_\vphi}{\tilde{g}}}_{0} = 0.\nonumber
\eeqa
The matrix $(h_{ij})$ of the (scalar) second fundamental form in the basis (\ref{basis}) is therefore
 \beqa
 \label{inducedh2}
    (h_{ij}) = \frac{1}{\sqrt{2\nu'F}}\left[
      \begin{array}{ccc}
       \nu{''}F + \nu'F_{r}(\nu'r_v - r_u) & 0 & 0\\
        0 & r\left(\nu'r_{v} - r_{u}\right) & 0\\
        0 & 0 & r \sin^2\vtheta\left(\nu'r_{v} - r_{u}\right)\\
      \end{array}
    \right].\nonumber\\
\eeqa
We can now write down Jang's equation.  It is $\mathfrak{g}^{ij}\mathfrak{h}_{ij} = 0$, where
\beqa
\mathfrak{g}_{ij} &=& g_{ij} + f_if_j,\nonumber\\
\mathfrak{h}_{ij} &=& h_{ij} -\frac{\nabla_i\nabla_j f}{\sqrt{1+|df|_{g}^2}},\nonumber
\eeqa
$\nabla$ is the Levi-Civita connection on the hypersurface $(M,g)$, $f(u)$ is the unknown function,
%$f(u,v)$ is the unknown function (which we identify with a function of a single variable, also denoted $f$, via $u \mapsto f(u,\nu(u))$),
and $g_{ij}$ and $h_{ij}$ are given by (\ref{inducedg2}) and (\ref{inducedh2}), respectively.  (We will write $f'$ in place of $f_u$.)   To determine $\mathfrak{g}^{ij}$, first note that
 \beqa
    (\mathfrak{g}_{ij}) = \left[
      \begin{array}{ccc}
       2\nu'F + (f')^2 & 0 & 0\\
        0 & r^2 & 0\\
        0 & 0 & r^2 \sin^2\vtheta\\
      \end{array}
    \right],\nonumber
\eeqa
%\vskip 9pt
because $f_{\vtheta} = f_{\vphi} \equiv 0$ (note also that $2\nu'F > 0$).  Since $|df|_{g}^2 = g^{uu}(f')^2 = (2\nu'F)^{-1}(f')^2$, the inverse metric is
\beqa
    (\mathfrak{g}^{ij}) = \left[
      \begin{array}{ccc}
       \frac{1}{2\nu'F+(f')^2} & 0 & 0\\
        0 & \frac{1}{r^2} & 0\\
        0 & 0 & \frac{1}{r^2 \sin^2\vtheta}\\
      \end{array}
    \right].\nonumber
\eeqa
%\vskip 9pt
To determine $\mathfrak{h}_{ij}$, first note that
\beqa
\nabla_u\nabla_u f \!&=&\! f'' - \Gamma^u_{uu}\,f' \ =\  f'' - (2\nu'F)^{-1}\left(\nu''F + \nu'F_r\,(r_u+r_v\nu')\right)f',\nonumber\\
\nabla_\vtheta\nabla_\vtheta f \!&=&\! -\Gamma^u_{\vtheta\vtheta}\,f' \ =\ (2\nu'F)^{-1}r\,(r_u+r_v\nu')\,f', \nonumber\\
\nabla_\vphi\nabla_\vphi f \!&=&\!-\Gamma^u_{\vphi\vphi}\,f' \ =\ (2\nu'F)^{-1}\sin^2\vtheta\,r\,(r_u+r_v\nu')\,f',\nonumber
\eeqa
%\vskip 9pt
and all other $\nabla_i\nabla_j f = 0$ (bear in mind that the coordinates on $M$ are $(u,\vtheta,\vphi)$ and the $\Gamma^{k}_{ij}$ are the Christoffel symbols corresponding to $M$; they are not the same as the Christoffel symbols $\widetilde{\Gamma}^{k}_{ij}$ for $\widetilde{M}$ above).  Hence $\mathfrak{h}_{ij}$ is given by
 \beqa
       \mathfrak{h}_{uu} &=& \frac{\nu{''}F + \nu'F_{r}(\nu'r_v - r_u)}{\sqrt{2\nu'F}} - \frac{f''-(2\nu'F)^{-1}(\nu''F + \nu'F_r\,(r_u+r_v\nu'))\,f'}{\sqrt{1+(2\nu'F)^{-1}(f')^2}}\,,\nonumber\\
        \mathfrak{h}_{\vtheta\vtheta} &=& \frac{r\left(\nu'r_{v} - r_{u}\right)}{\sqrt{2\nu'F}} - \frac{(2\nu'F)^{-1}r\,(r_u+r_v\nu')\,f'}{\sqrt{1+(2\nu'F)^{-1}(f')^2}}\,,\nonumber\\
        \mathfrak{h}_{\vphi\vphi} &=& \frac{r \sin^2\vtheta\left(\nu'r_{v} - r_{u}\right)}{\sqrt{2\nu'F}} - \frac{(2\nu'F)^{-1}\sin^2\vtheta\,r\,(r_u+r_v\nu')\,f'}{\sqrt{1+(2\nu'F)^{-1}(f')^2}}\,,\nonumber
\eeqa
%\vskip 9pt
and all other $\mathfrak{h}_{ij} = 0$.  Jang's equation is thus the sum of the following three terms set equal to zero:
\beqa
\mathfrak{g}^{uu}\mathfrak{h}_{uu} &=&\frac{1}{2\nu'F + (f')^2}\left[\frac{\nu{''}F + \nu'F_{r}(\nu'r_v - r_u) }{\sqrt{2\nu'F}}\right.\nonumber\\
&&\left.-\,\frac{f'' - (2\nu'F)^{-1}(\nu''F+\nu'F_r\,(r_u+r_v\nu'))\,f'}{\sqrt{1+(2\nu'F)^{-1}(f')^2}}\right],\nonumber\\
\mathfrak{g}^{\vtheta\vtheta}\mathfrak{h}_{\vtheta\vtheta} &=&\frac{1}{\sqrt{2\nu'F}}\left[\frac{\left(\nu'r_{v} - r_{u}\right)}{r} - \frac{(r_u+r_v\nu')\,f'}{r\sqrt{(2\nu'F)+(f')^2}}\right],\nonumber\\
\mathfrak{g}^{\vphi\vphi}\mathfrak{h}_{\vphi\vphi} &=& \frac{1}{\sqrt{2\nu'F}}\left[\frac{\left(\nu'r_{v} - r_{u}\right)}{r} - \frac{(r_u+r_v\nu')\,f'}{r\sqrt{(2\nu'F)+(f')^2}}\right]\cdot\nonumber
\eeqa
%\vskip 9pt
Writing out this sum, we have, after a little simplification,
 \beqa
 \label{jangprefinal}
f'' \!\!\!&+&\!\!\! \left[(2\nu'F) + (f')^2\right]\frac{r_u+r_v\nu'}{r\nu'F}\,f' - \frac{\nu''F+\nu'F_r\,(r_u+r_v\nu')}{2\nu'F}\,f'\nonumber\\
&&-\,\sqrt{(2\nu'F) + (f')^2}\,\,\frac{\nu{''}F + \nu'F_{r}(\nu'r_v - r_u)}{2\nu'F}\\
&&-\, \left[(2\nu'F) + (f')^2\right]^{3/2}\frac{\left(\nu'r_{v} - r_{u}\right)}{r\nu'F}\ =\ 0.\nonumber
 \eeqa
% \vskip 9pt
 For the hypersurfaces under consideration in this paper, $\nu(u) = u + c$ with $c \in (-\sqrt{8m/e},\sqrt{8m/e})$.  Furthermore, we have the following identities:
 \beqa
 \label{rs}
r_u\ =\ \frac{\nu(u)}{\phi_r(r)}
\hspace{.2in},\hspace{.2in}r_v\ =\ \frac{u}{\phi_r(r)}\hspace{.2in},\hspace{.2in}F_r(r)\ =\ 
-\left(\frac{1}{r}+\frac{1}{2m}\right)F(r).\nonumber
\eeqa
%Inserting these into eqn.~(\ref{jangprefinal}), we obtain
% \beqa
%f'' \!\!\!&+&\!\!\! \left[2F(r) + (f')^2\right]\frac{u+c}{r\phi_r(r)F(r)}\,f' +\left(\frac{1}{r}+\frac{1}{2m}\right)\frac{u+c}{2\phi_r(r)}f' -\nonumber\\
%&&\sqrt{2F(r) + (f')^2}\left(\frac{1}{r}+\frac{1}{2m}\right)\frac{c}{2\phi_r(r)} + \left[2F(r) + (f')^2\right]^{3/2}\frac{c}{r\phi_r(r)F(r)} = 0,\nonumber
% \eeqa 
Inserting these into (\ref{jangprefinal}), we obtain
 \beqa
&&f''\ +\ \underbrace{\left[2F(r) + (f')^2\right]\frac{2u+c}{r\phi_r(r)F(r)}\,f' +\left(\frac{1}{r}+\frac{1}{2m}\right)\frac{2u+c}{2\phi_r(r)}f'}_{\frac{2u+c}{2\phi_r(r)}\left(\frac{1}{2m}\,+\,\frac{5}{r}\right)f'\ +\ \frac{2u+c}{r\phi_r(r)F(r)}(f')^3}\nonumber\\
&&+\ \underbrace{\left[2F(r) + (f')^2\right]^{3/2}\frac{c}{r\phi_r(r)F(r)}-\sqrt{2F(r) + (f')^2}\left(\frac{1}{r}+\frac{1}{2m}\right)\frac{c}{2\phi_r(r)}}_{-\sqrt{\frac{F(r)}{2}+\frac{(f_c')^2}{4}}\frac{c}{\phi_r(r)}\left[\frac{1}{2m}
\,-\,\frac{3}{r}\,-\,\frac{2(f_c')^2}{rF(r)}\right]}\ =\ 0,\nonumber
 \eeqa 
which is precisely (\ref{jang}), once we note that $F(r)\phi_r(r) = 4m$.

\bibliographystyle{unsrt}
\bibliography{Jang2}

\begin{thebibliography}{10}

\bibitem{J78}
Pong~Soo Jang.
\newblock On the positivity of energy in general relativity.
\newblock {\em J. Math. Phys.}, 19(5):1152--1155, 1978.

\bibitem{SY81}
Richard Schoen and Shing~Tung Yau.
\newblock Proof of the positive mass theorem. {II}.
\newblock {\em Comm. Math. Phys.}, 79(2):231--260, 1981.

\bibitem{SY83}
Richard Schoen and S.~T. Yau.
\newblock The existence of a black hole due to condensation of matter.
\newblock {\em Comm. Math. Phys.}, 90(4):575--579, 1983.

\bibitem{E95}
Douglas~M. Eardley.
\newblock Gravitational collapse of vacuum gravitational field configurations.
\newblock {\em J. Math. Phys.}, 36(6):3004--3011, 1995.

\bibitem{Y01}
Shing~Tung Yau.
\newblock Geometry of three manifolds and existence of black hole due to
  boundary effect.
\newblock {\em Adv. Theor. Math. Phys.}, 5(4):755--767, 2001.

\bibitem{AEM10}
Lars Andersson, Michael Eichmair, and Jan Metzger.
\newblock Jang's equation and its applications to marginally trapped surfaces.
\newblock In {\em Complex analysis and dynamical systems {IV}. {P}art 2},
  volume 554 of {\em Contemp. Math.}, pages 13--45. Amer. Math. Soc.,
  Providence, RI, 2011.

\bibitem{AM09}
Lars Andersson and Jan Metzger.
\newblock The area of horizons and the trapped region.
\newblock {\em Comm. Math. Phys.}, 290(3):941--972, 2009.

\bibitem{Eichmair}
Michael Eichmair.
\newblock The {P}lateau problem for marginally outer trapped surfaces.
\newblock {\em J. Differential Geom.}, 83(3):551--583, 2009.

\bibitem{EM12}
Michael Eichmair and Jan Metzger.
\newblock Jenkins--{S}errin type results for the {J}ang equation.
\newblock preprint: http://arxiv.org/abs/1205.4301v1.

\bibitem{MalecMurchadha}
Edward Malec and Niall~O Murchadha.
\newblock The {J}ang equation, apparent horizons and the {P}enrose inequality.
\newblock {\em Classical and Quantum Gravity}, 21(24):5777, 2004.

\bibitem{P65}
Roger Penrose.
\newblock Gravitational collapse and space-time singularities.
\newblock {\em Phys. Rev. Lett.}, 14:57--59, 1965.

\bibitem{eichmair2012topological}
Michael Eichmair, Gregory~J Galloway, and Daniel Pollack.
\newblock Topological censorship from the initial data point of view.
\newblock {\em J. Differential Geom.}, 95(3):389--405, 2013.

\bibitem{W84}
Robert~M. Wald.
\newblock {\em General relativity}.
\newblock University of Chicago Press, Chicago, IL, 1984.

\bibitem{WI91}
Robert~M. Wald and Vivek Iyer.
\newblock Trapped surfaces in the {S}chwarzschild geometry and cosmic
  censorship.
\newblock {\em Phys. Rev. D (3)}, 44(12):R3719--R3722, 1991.

\bibitem{T94}
K.~P. Tod.
\newblock Conical singularities and torsion.
\newblock {\em Classical and Quantum Gravity}, 11(5):1331--1339, 1994.

\bibitem{P79}
Roger Penrose.
\newblock Singularities and time-asymmetry.
\newblock In {\em General Relativity: An {E}instein Centenerary Survey}, pages
  581--638. Cambridge University Press, Cambridge, 1979.

\bibitem{T87}
K.~P. Tod.
\newblock Quasilocal mass and cosmological singularities.
\newblock {\em Classical and Quantum Gravity}, 4(5):1457--1468, 1987.

\bibitem{G91}
Stephen~W. Goode.
\newblock Isotropic singularities and the {P}enrose-{W}eyl tensor hypothesis.
\newblock {\em Classical Quantum Gravity}, 8(1):L1--L6, 1991.

\bibitem{ON83}
Barrett O'Neill.
\newblock {\em Semi-{R}iemannian geometry}, volume 103 of {\em Pure and Applied
  Mathematics}.
\newblock Academic Press Inc. [Harcourt Brace Jovanovich Publishers], New York,
  1983.
\newblock With applications to relativity.

\bibitem{W10}
Catherine Williams.
\newblock On blow-up solutions of the {J}ang equation in spherical symmetry.
\newblock {\em Classical and Quantum Gravity}, 27(6):065001, 13, 2010.

\end{thebibliography}

\end{document}